\documentclass[a4paper,11pt]{article}

\usepackage[margin=2.5cm]{geometry}
\usepackage[utf8]{inputenc}

\usepackage{amsmath}
\usepackage{amsfonts}
\usepackage{amssymb}
\usepackage{lipsum}
\usepackage{amsthm}
\usepackage{hyperref}

\usepackage[sort&compress,nameinlink,noabbrev,capitalize]{cleveref}
\crefname{algorithm}{Algorithm}{Algorithms}

\usepackage{booktabs}
\usepackage{tabularx}
\usepackage{multirow}

\usepackage[citestyle=numeric,
            bibstyle=numeric,
            hyperref=true,
            url=false,
            isbn=false,
            backref=true,
            backrefstyle=three,
            maxcitenames=3,
            maxbibnames=100,
            block=none,
            bibencoding=utf8,
            backend=bibtex]{biblatex}

\usepackage{graphicx}
\graphicspath{{graphics/}}

\theoremstyle{plain}
\newtheorem{theorem}{Theorem}[section]
\newtheorem{lemma}[theorem]{Lemma}

\theoremstyle{definition}

\usepackage{authblk}

\newcommand{\orcidID}[1]{\href{https://orcid.org/#1}{\includegraphics[scale=.03]{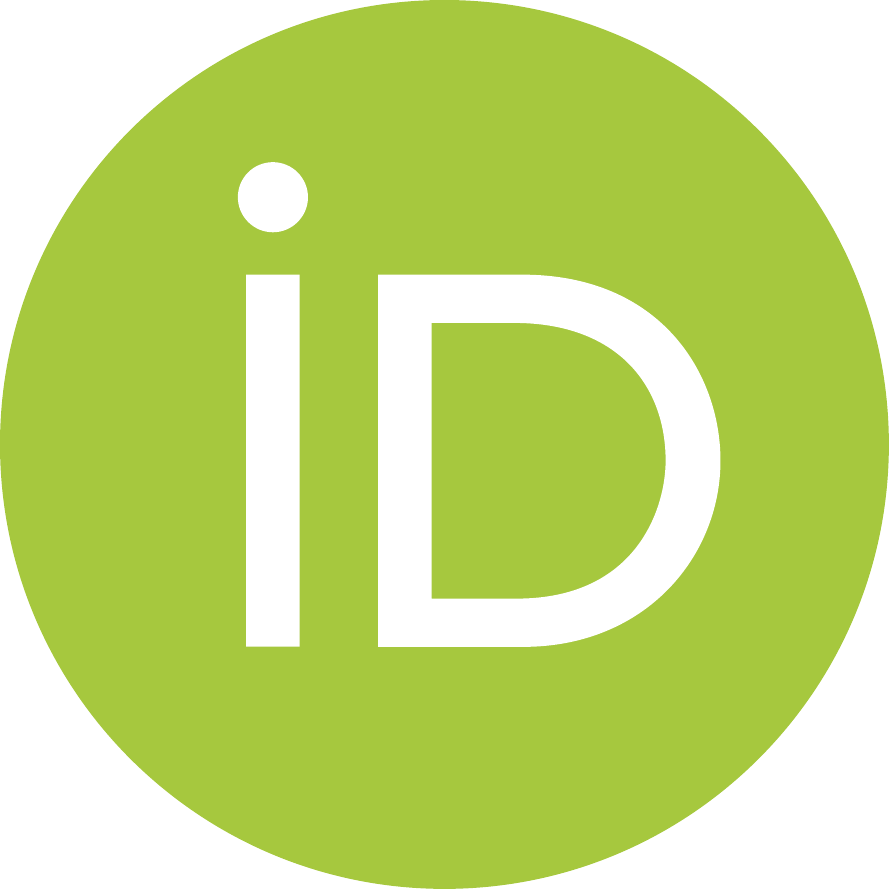}}} 
\usepackage{academicons}
\usepackage{xcolor}
\usepackage{fontawesome5}
\usepackage[basic]{complexity}
\usepackage{subcaption}
\def\mailsymbol{\fontsize{9}{12}\sffamily\bfseries \faIcon[regular]{envelope}}
\newcommand\mail[1]{\href{mailto:#1}{\mailsymbol}}
\newcommand{\oscm}{\textsc{One-Sided Crossing Minimization}}
\newcommand{\oscmshort}{\textsc{OSCM}}
\newcommand{\olcm}{\textsc{One-Layer Crossing Minimization}}
\newcommand{\olcmshort}{\textsc{OLCM}}
\title{A Note on the Complexity of One-Sided Crossing Minimization of Trees}

\author{Alexander Dobler \orcidID{0000-0002-0712-9726} \mail{adobler@ac.tuwien.ac.at}}

\affil{Algorithms and Complexity Group, TU Wien, Austria}

\usepackage{scrextend}          %
\usepackage{needspace}          %
\usepackage[ruled,linesnumbered]{algorithm2e} %
\usepackage[textsize=footnotesize]{todonotes}
\usepackage{xcolor,colortbl}
\usepackage{MnSymbol}

\begin{document}
\maketitle
\begin{abstract}
  In 2011, Harrigan and Healy published a polynomial-time algorithm for one-sided crossing minimization for trees \cite{DBLP:conf/walcom/HarriganH11}. We point out a counterexample to that algorithm, and show that one-sided crossing minimization is \NP-hard for trees.
\end{abstract}

\pagenumbering{arabic}
\section{Introduction}
One-sided crossing minimization deals with the minimization of crossings in drawings of bipartite graphs. Given a bipartite graph $G=(V_1\cupdot V_2,E)$ with the vertex set $V_1$ ordered vertically on the left side by a permutation $\pi_1$, the task is to find a vertical order $\pi_2$ of $V_2$ on the right side to minimize the number of pairwise crossings between the edges in $E$ if they are drawn as straight lines (see \cref{fig:bipartite}). Here, $A\cupdot B$ is the disjoint union of two sets $A$ and $B$. In this paper, we call such a drawing \emph{two-layer drawing} of $G$ respecting $\pi_1$ and $\pi_2$. The problem is known as \oscm\ (\oscmshort), or as \olcm\ (\olcmshort).
It is well-known that this problem is \NP-hard for general bipartite graphs \cite{DBLP:journals/algorithmica/EadesW94}. A 2011 WALCOM paper \cite{DBLP:conf/walcom/HarriganH11} presented a polynomial-time algorithm for the case when $G$ is a tree. In \cref{sec:counterexample} we discuss their proof and give a counterexample to the statements leading to the presented algorithm. In \cref{sec:hardness} a proof for \NP-hardness of \oscm\ in trees is given.
\begin{figure}[tb]
  \centering
  \begin{subfigure}[t]{.3\linewidth}
    \centering
    \includegraphics{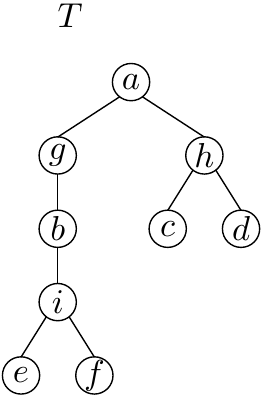}
    \caption{A tree $T$ (which forms a bipartite graph).}
    \label{fig:tree}
  \end{subfigure}
  \begin{subfigure}[t]{.3\linewidth}
    \centering
    \includegraphics{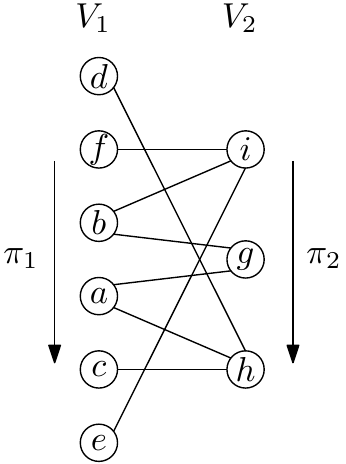}
    \caption{A two-layer drawing of $T$ with $9$ crossings.}
    \label{fig:bipartite}
  \end{subfigure}
  \begin{subfigure}[t]{.3\linewidth}
    \centering
    \includegraphics{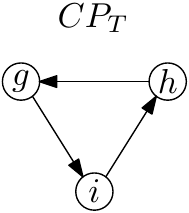}
    \caption{The directed graph $CP_T$.}
    \label{fig:cpt}
  \end{subfigure}
  \caption{A counterexample to $CP_T$ being acyclic.}
  \label{fig:counterexample}
\end{figure}

\section{The Proof in \cite{DBLP:conf/walcom/HarriganH11} and a Counterexample}\label{sec:counterexample}
The proof applies a well-known characterization of optimal solutions of \oscmshort\ in terms of feedback arc set.
Namely, let us consider an instance $G=(V_1\cupdot V_2,E)$ of \oscmshort\ with a fixed permutation $\pi_1$ of $V_1$. For $u,v\in V_2$, let $c(u,v)$ be the number of crossings between edge pairs in $N_G(u)\times N_G(v)$ if $u$ is placed before $v$ in permutation $\pi_2$, where $N_G(u)$ (resp.\ $N_G(v)$) are the edges incident to $u$ (resp.\ $v$) in $G$. Let $CP_G$ be the directed multi-graph defined on the vertex set $V_2$ such that, for $u,v\in V_2$, there exist $cr(u,v)-cr(v,u)$ edges from $u$ to $v$ if $cr(u,v)>cr(v,u)$. The number of crossings in an optimal solution can then be obtained by the cardinality of a minimum feedback arc set $A$ in $CP_G$, and the ordering $\pi_2$ achieving this amount of crossings can be obtained by removing $A$ from $CP_G$ and defining $\pi_2$ as some topological ordering of the obtained directed acyclic graph \cite{DBLP:journals/tsmc/SugiyamaTT81}.

The authors in \cite{DBLP:conf/walcom/HarriganH11} claim that $CP_T$ is acyclic for trees $T$. This would imply a simple algorithm for \oscm\ by computing a topological ordering of $CP_T$. In \cref{fig:counterexample} we give a counterexample to that claim. Namely, let $T$ be as in \cref{fig:tree}, such that $V_1=\{a,b,c,d,e,f\},V_2=\{g,h,i\}$ and let $\pi_1=(d,f,b,a,c,e)$.
This implies that
\begin{align*}
  \text{cr}(g,h)=2,&\text{cr}(h,g)=3,\\
  \text{cr}(g,i)=3,&\text{cr}(i,g)=2,\\
  \text{cr}(h,i)=4,&\text{cr}(i,h)=5,
\end{align*}
which results in the graph $CP_T$ in \cref{fig:cpt} that is not acyclic. In fact, this is one of the counterexamples s.t.\ $T$ has the fewest number of vertices.
\section{NP-hardness}\label{sec:hardness}
\oscmshort\ is \NP-hard for trees by a reduction from another restricted variant of \oscmshort\ discussed in \cite{DBLP:conf/gd/MunozUV01}: it was shown that \oscmshort\ is \NP-hard for bipartite graphs $G=(V_1\cupdot V_2,E)$ such that $G$ is the disjoint union of 4-stars, that is, the connected components of $G$ are 5-vertex graphs with one vertex having degree 4 and the remaining vertices having degree 1. Furthermore, $V_1$ consists of all the degree-1 vertices, while $V_2$ consists of all the degree-4 vertices. Given such a graph $G$ and a permutation $\pi_1$ of $V_1$, it is \NP-hard to find a permutation $\pi_2$ of $V_2$ that minimizes the number of edge crossings. 
\begin{figure}[t]
  \centering
  \includegraphics{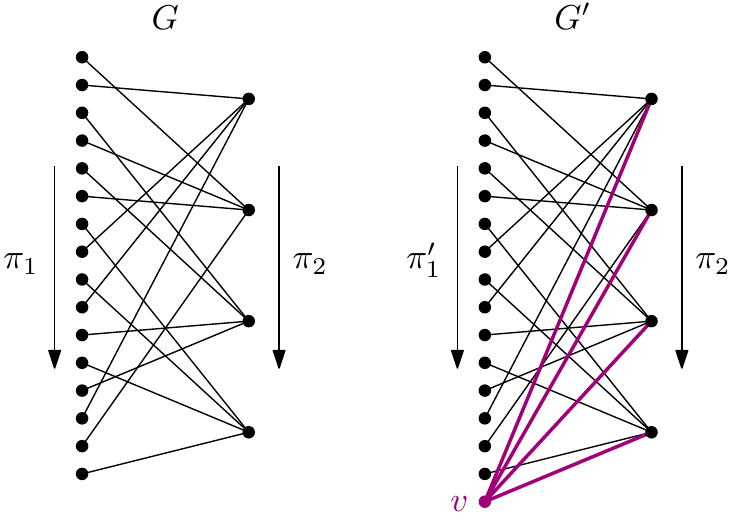}
  \caption{Sketch of the construction of $G'$.}
  \label{fig:sketch}
\end{figure}
Let us consider the graph $G'$ that is constructed from $G$ by adding a new vertex $v$ to $V_1$ and connecting $v$ to all vertices in $V_2$ (see \cref{fig:sketch}). Note that $G'$ is a tree and consider the ordering $\pi_1'$ of $V_1(G')$ that is constructed from $\pi_1$ by appending $v$ to the end.
Let $\pi_2$ be a permutation of $V_2$.
The following is easy to see.
\begin{lemma}
  A two-layer drawing of $G$ respecting $\pi_1$ and $\pi_2$ has $k$ crossings if and only if a two-layer drawing of $G'$ respecting $\pi_1'$ and $\pi_2$ has $k+|V_2|(|V_2|-1)$ crossings.
\end{lemma}
\begin{proof}
  We calculate the number of crossings involving edges incident to $v$ (the purple edges in \cref{fig:sketch}) in the drawing of $G'$ as follows.
  \begin{align*}
    4+2\cdot 4+3\cdot 4+\dots+(|V_2|-1)\cdot 4=2|V_2|(|V_2|-1).
  \end{align*}
  The remaining crossings in $G'$ correspond one-to-one to crossings in $G$.
\end{proof}
\begin{theorem}
  \oscm\ is \NP-hard for trees.
\end{theorem}

\printbibliography{}

\end{document}